\documentclass[a4paper]{article}

\usepackage{amsmath}
\usepackage{amssymb}
\usepackage[dvips]{graphicx}
\usepackage[OT4]{fontenc}
\usepackage[latin2]{inputenc}
\usepackage{url}
\usepackage[numbers]{natbib}
\usepackage{bm}
\usepackage{setspace}

\newcommand{\abs}[1]{\left| #1 \right|}
\newcommand{\norm}[1]{\left|\mkern-1mu\left| #1 \right|\mkern-1mu\right|}

\newcommand{\okra}[1]{\left( #1 \right)}
\newcommand{\kwad}[1]{\left[ #1 \right]}
\newcommand{\klam}[1]{\left\{ #1 \right\}}
\newcommand{\ceil}[1]{\left\lceil #1 \right\rceil}

\DeclareMathOperator{\sred}{\mathbf{E}}

\newcommand{\boole}[1]{{\bf 1}_{\klam{#1}}}

\newtheorem{proposition}{Proposition}

\newtheorem{lemma}{Lemma}
\newenvironment*{proof}{\begin{trivlist}\item[]
\noindent\textbf{Proof:}}{$\Box$\par\end{trivlist}}
\newenvironment*{proof*}[1]{\begin{trivlist}\item[]
\noindent\textbf{Proof of #1:}}{$\Box$\par\end{trivlist}}

\renewcommand{\thesection}{\Roman{section}}

\title{Mixing, Ergodic, and Nonergodic Processes with \\ Rapidly
  Growing Information between Blocks} \date{}

\author{{\L}ukasz D\k{e}bowski%
  \thanks{
    The research reported in Section \ref{secExcess} of this work was
    supported by the IST Programme of the European Community, under
    the PASCAL II Network of Excellence, IST-2002-506778. %
    \newline\null\hspace{\parindent}
    {\L}. D\k{e}bowski is with
    the Institute of Computer Science, Polish Academy of Sciences, ul.
    Ordona 21, 01-237 Warszawa, Poland (e-mail: ldebowsk@ipipan.waw.pl).}
}

\begin{document}

\pagestyle{empty}   
\begin{titlepage}
\maketitle

\begin{abstract}
  We construct mixing processes over an infinite alphabet and ergodic
  processes over a~finite alphabet for which Shannon mutual
  information between adjacent blocks of length $n$ grows as
  $n^\beta$, where $\beta\in(0,1)$. The processes are a~modification
  of nonergodic Santa Fe processes, which were introduced in the
  context of natural language modeling. The rates of mutual
  information for the latter processes are alike and also established
  in this paper. As an auxiliary result, it is shown that infinite
  direct products of mixing 
  processes are also mixing. 
  \\[1em]
  \\[1em]
  \\[1em]
  \textbf{Running head}: Processes with Rapidly Growing 
  Information
\end{abstract}

\end{titlepage}
\pagestyle{plain}   


\section{Introduction}
\label{secIntro}

Let $H(X):=\sred\kwad{-\log P(X)}$ denote the entropy of
a~discrete variable $X$ on a~probability space
$(\Omega,\mathcal{J},P)$, where $\sred$ is the expectation with
respect to $P$, $\log$ is the binary logarithm, and the variable
$P(X)$ takes the value $P(X=x)$ for $X=x$. We have the mutual
information $I(X;Y):=H(X)+H(Y)-H(X,Y)$ for finite entropies on
the right hand side. Besides, we have the conditional entropy
$H(X|Z)=H(X,Z)-H(Z)$ and the conditional mutual information
$I(X;Y|Z):=H(X|Z)+H(Y|Z)-H(X,Y|Z)$. These definitions are
generalized to arbitrary random variables, e.g., in
\cite{Pinsker60en,Debowski09}.

Let $(X_i)_{i\in\mathbb{Z}}$ be a~stationary process on
$(\Omega,\mathcal{J},P)$, where $X_i:(\Omega,\mathcal{J})\rightarrow
(\mathbb{X},\mathcal{X})$. For its distribution
$\mu=P((X_i)_{i\in\mathbb{Z}}\in\cdot)$ we denote the mutual
information between blocks of length $n$ as
\begin{align}
  \label{En}
  E_\mu(n):=I\okra{X_{1:n};X_{n+1:2n}} 
  .
\end{align}
The limiting value of mutual information, called excess entropy, is
defined as
\begin{align}
  E_\mu:=I((X_i)_{i\le 0};(X_i)_{i\ge 1})=
  \lim_{n\rightarrow\infty} E_\mu(n)  
\end{align}
These quantities are natural measures of dependence for
discrete-valued processes \cite{CrutchfieldFeldman03}.
We are interested in constructing diverse examples of stationary
measures for which
\begin{align}
  \label{limEn}
  E_\mu(n)\asymp n^\beta
  ,
\end{align}
where $\beta\in(0,1)$, because certain measures of this kind may be
useful for modeling natural language, cf.,
\cite{Hilberg90,Debowski11b}.

Mentioning related results, let us first consider Gaussian
processes. For theses processes the conditional mutual information
equals $I(X_0;X_n|(X_i)_{i=1}^{n-1})=-\log(1-|\alpha(n)|^2)$, where
function $\alpha(k)$ is the partial autocorrelation, cf.,
\cite{BrockwellDavis87}. Regardless of the alphabet, the mutual
information between blocks may be reconstructed from conditional
mutual information as
\begin{align}
  E_\mu(n)=
  \sum_{k=1}^{n-1} k I(X_0;X_k|(X_i)_{i=1}^{k-1})
  +
  \sum_{k=n}^{2n-1} (2n-k) I(X_0;X_k|(X_i)_{i=1}^{k-1})
  .
\end{align}
Thus the asymptotics (\ref{limEn}) holds if and only if
$\sum_{k=1}^nk\abs{\alpha(k)}^2\asymp n^\beta$.  As a~result, the
construction of processes that satisfy condition (\ref{limEn}) is easy
because the sole constraint on partial correlation reads
$\abs{\alpha(k)}\le 1$ \cite{Ramsey74}. However, a~classical result
\cite{Finch60} says that excess entropy of nonsingular Gaussian
autoregressive moving average (ARMA) processes is finite, cf.,
\cite{CrutchfieldFeldman03}, \cite[Theorem 9.4.1]{CoverThomas91},
\cite[Section 5.5]{GrenanderSzego84}.

Some examples of stationary processes for which excess entropy is
infinite are also known for discrete-valued processes. The trivial
example for a~countably infinite alphabet is a~process such that
$X_i$ does not depend on $i$ and $H(X_i)=\infty$. Then we have
$E_\mu(n)=\infty$ for any $n\ge 1$.
The aforementioned construction is impossible for processes over
a~finite alphabet.  Considering those processes, we mention first that
asymptotics $E(n)= (k/2)\log (n/2\pi e)+O(1)$ holds for any Bayesian
mixture of a~$k$-parameter model with a~prior concentrated on a~subset
of parameters with bounded Fisher information \cite[Theorem
8.3]{Grunwald07}.
Similar asymptotics $E(n)\asymp \log n$ holds for a~binary process
constructed by Gramss \cite{Gramss94}, cf.,
\cite{CrutchfieldYoung89,Ebeling97}. The distribution of that process
is formed by the frequencies of 0's and 1's in the rabbit sequence.
As for processes with infinite excess entropy that are mixing, Bradley
\cite{Bradley80} constructed a~binary process which satisfies two
conditions, cf., \cite{Sujan83}: (i) the process is $\rho$-mixing and
(ii) the restricted measure $P((X_i)_{i\le 0\lor i\ge n}\in \cdot)$
is singular with respect to the product measure $P((X_i)_{i\le 0}\in
\cdot) \times P((X_i)_{i\ge n}\in \cdot)$ for any $n\ge 1$
\cite[Lemma 3]{Bradley80}. The first property implies that the process
is mixing in the ordinary ergodic theoretic sense \cite[Volume 1,
Chapters 3 and 5]{Bradley07}. The second property implies that the
excess entropy is infinite.

A~few other examples concern hidden Markov chains. By the data
processing inequality, excess entropy is finite for hidden Markov
chains with a~finite number of hidden states
\cite{ShaliziCrutchfield01}.  On the other hand, if the distribution
of ergodic components of a~stationary process has infinite entropy
then the process has infinite excess entropy \cite[Theorem
5]{Debowski09}. Such a~situation may arise for hidden Markov chains
with a~countably infinite number of hidden states.  (Consider for
instance a~mixture of periodic processes where the probability of
a~period is a~sufficiently slowly decreasing function of the cycle
length \cite{TraversCrutchfield11}.)
A~less trivial example, constructed in \cite{TraversCrutchfield11}, is
a~stationary \emph{ergodic} hidden Markov chain with infinite excess
entropy, a~finite number of output symbols, and a~countably infinite
alphabet of hidden states.

In this paper we will consider another class of processes that are
nonergodic, ergodic, or mixing and satisfy condition
(\ref{limEn}). The construction of these processes is motivated
linguistically. Let us first sketch this motivation.  In our previous
work \cite{Debowski11b}, we have shown that proportionality
(\ref{limEn}) implies a~power law which resembles Zipf's law for the
distribution of words.  Namely, product $E_\mu(n)\log n$ is upper
bounded by the expected vocabulary size of an admissibly minimal
grammar for the text of length $n$. It was empirically observed that
the latter quantity approximates the number of distinct words for
texts in natural language \cite{DeMarcken96}. Our bound for mutual
information and the vocabulary size holds if the alphabet $\mathbb{X}$
is finite and the process's distribution has finite energy property
\cite[Theorem 3]{Debowski11b}. There is also another linguistically
motivated bound for $E_\mu(n)$. That one is a~lower bound. Namely,
asymptotics
\begin{align}
  \label{limsupEn}
  \limsup_{n\rightarrow\infty} E_\mu(n)/n^\beta>0
\end{align}
follows from a~hypothesis that texts describe an infinite random
object in a~highly repetitive way so that $n^\beta$ independent facts
about the object can be inferred on average from the text of length
$n$ \cite[Theorem 2]{Debowski11b}.

The goal of this paper is to prove the stronger asymptotics
(\ref{limEn}) for processes that were discussed in \cite{Debowski11b}
and to define a~new model of texts that describe a~random object.  So
far, we have considered objects that do not change in time.  This
leads to models of texts being nonergodic measures. Here, we will
admit objects that evolve slowly. That leads to models of texts which
are mixing measures and still satisfy proportionality
(\ref{limEn}). In this way, linguistic inspiration contributes to
better understanding of yet another problem in information theory.

Let us introduce our basic example.  Throughout this paper,
$(X_i)_{i\in\mathbb{Z}}$ denotes a~stationary process on
$(\Omega,\mathcal{J},P)$ with $X_i:(\Omega,\mathcal{J})\rightarrow
(\mathbb{X},\mathcal{X})$ and $\mathbb{X}=\mathbb{N}\times\klam{0,1}$,
where $\mathbb{N}$ is the set of positive integers.  In the series of
papers \cite{Debowski09,Debowski10,Debowski11b} we have examined some
properties of the following process $(X_i)_{i\in\mathbb{Z}}$, called
the (original) Santa Fe process in \cite{Debowski11b}. Namely, the
variables $X_i$ consist of pairs
\begin{align}
  \label{exUDP}
  X_i&=(K_i,Z_{K_i})
  ,
\end{align}
where processes $(K_i)_{i\in\mathbb{Z}}$ and $(Z_k)_{k\in\mathbb{N}}$
are independent and distributed as follows. First, variables $Z_k$
are binary and equidistributed,
\begin{align}
  P(Z_k=0)=P(Z_k=1)&=1/2
  ,
  &
  (Z_k)_{k\in\mathbb{N}}&\sim \text{IID}
  .
\end{align}
Second, variables $K_i$ obey the power law
\begin{align}
  \label{ZetaK}
  P(K_i=k)&=k^{-1/\beta}/\zeta(\beta^{-1})
  , 
  &
  (K_i)_{i\in\mathbb{Z}}&\sim \text{IID}
  ,
\end{align}
where $\beta\in(0,1)$ and $\zeta(x)=\sum_{k=1}^\infty k^{-x}$ is the
zeta function. 

Let us recall that $\mu=P((X_i)_{i\in\mathbb{Z}}\in\cdot)$ and
$E_\mu(n)=I\okra{X_{1:n};X_{n+1:2n}}$. The first new result of this
paper is:
\begin{proposition}
  \label{theoEnUDP}
  The block mutual information $E_\mu(n)$ for the original Santa Fe
  process $(X_i)_{i\in\mathbb{Z}}$ given by formula (\ref{exUDP})
  obeys
  \begin{align}
    \label{limEnUDP}
    \lim_{n\rightarrow\infty} \frac{E_\mu(n)}{n^\beta}=
    \frac{(2-2^\beta)\Gamma(1-\beta)}{[\zeta(\beta^{-1})]^\beta}
    .
  \end{align}
\end{proposition}
The calculation of the limit
is facilitated by a~decomposition of mutual information between
blocks $X_{1:n}$ and $X_{n+1:2n}$ into a~series of triple information
among blocks $X_{1:n}$ and $X_{n+1:2n}$ and variables $Z_k$. This
decomposition is a~particular property of the Santa Fe process and
some similar measures.

The uncommon construction of process (\ref{exUDP}) can be interpreted
in this way. Imagine that the Santa Fe process is a~sequence of
statements which describe a~random object $(Z_k)_{k\in\mathbb{N}}$
consistently. Each statement $X_i=(k,z)$ reveals both the address $k$
of a~random bit of $(Z_k)_{k\in\mathbb{N}}$ and its value $Z_k=z$.
Observe that the description is repetitive and consistent: if two
statements $X_i=(k,z)$ and $X_j=(k',z')$ describe bits of the same
address ($k=k'$) then they always assert the same bit value ($z=z'$).
It follows hence that variables $Z_k$ can be predicted from
realization $(X_i)_{i\in\mathbb{Z}}$ in a~shift-invariant way and
therefore the Santa Fe process is (strongly) nonergodic, cf.,
\cite{Debowski09}, \cite[Definition 1]{Debowski11b}.

Now let us introduce an example of a mixing process which satisfies
(\ref{limEn}).  For this goal, we will replace individual variables
$Z_k$ in the Santa Fe process with Markov chains
$(Z_{ik})_{i\in\mathbb{Z}}$. These Markov chains will be obtained by
iterating a~binary symmetric channel.
Subsequently, the following process $(X_i)_{i\in\mathbb{Z}}$ will be
called the generalized Santa Fe process. Let us put
\begin{align}
  \label{exMixing}
  X_i&=(K_i,Z_{i,K_i})
  ,
\end{align}
where processes $(K_i)_{i\in\mathbb{Z}}$ and
$(Z_{ik})_{i\in\mathbb{Z}}$, where $k\in\mathbb{N}$, are independent
and distributed as follows.  First, variables $K_i$ are distributed
according to formula (\ref{ZetaK}), as before. Second, each process
$(Z_{ik})_{i\in\mathbb{Z}}$ is a~Markov chain with marginal
distribution
\begin{align}
  P(Z_{ik}=0)=P(Z_{ik}=1)&=1/2
\end{align}
and cross-over probabilities
\begin{align}
  P(Z_{ik}=0|Z_{i-1,k}=1)=P(Z_{ik}=1|Z_{i-1,k}=0)&=p_k
  .
\end{align}

A~linguistic interpretation of this process is as follows.  Facts that
are mentioned in texts repeatedly fall roughly under two types, as
mentioned in the discussion of Definition 1 in \cite{Debowski11b}: (i)
facts about objects that do not change in time (like mathematical or
physical constants), and (ii) facts about objects that evolve with
a~varied speed (like culture, language, or geography). The random
object $(Z_{k})_{k\in\mathbb{N}}$ described by the original Santa Fe
process does not evolve, or rather, no bit $Z_{k}$ is ever forgotten
once revealed. On the other hand, the object
$(Z_{ik})_{k\in\mathbb{N}}$ described by the generalized Santa Fe
process is a~function of an instant $i$ and the probability that the
$k$-th bit flips at a~given instant equals $p_k$.  For vanishing
cross-over probabilities, the generalized Santa Fe process collapses
to the original process.

As we will establish later in this paper, the generalized Santa Fe
process is mixing for cross-over probabilities different to $0$ or
$1$.
\begin{proposition}
  \label{theoMixing}
  The generalized Santa Fe process $(X_{i})_{i\in\mathbb{Z}}$ given by
  formula (\ref{exMixing}) is mixing for $p_k\in(0,1)$.
\end{proposition}
The proof consists in noticing that infinite direct products of mixing
processes are mixing. This is an easy generalization of the well known
fact for finite products \cite[Chapter 10.\S
1]{CornfeldFominSinai82en}.  

We will also demonstrate this fact, which generalizes Proposition
\ref{theoEnUDP}:
\begin{proposition}
  \label{theoEnMixing}
  The block mutual information $E_\mu(n)$ for the generalized Santa Fe
  process $(X_i)_{i\in\mathbb{Z}}$ given by formula (\ref{exMixing})
  obeys
  \begin{align}
    \label{limEnMixing}
    \limsup_{n\rightarrow\infty} \frac{E_\mu(n)}{n^\beta}\le
    \frac{(2-2^\beta)\Gamma(1-\beta)}{[\zeta(\beta^{-1})]^\beta}
    .
  \end{align}
  The lower limits in particular cases are as follows:
  \begin{enumerate}
  \item If $p_k\le P(K_i=k)$ then
    \begin{align}
      \label{limEnMixingPower}
      \liminf_{n\rightarrow\infty} \frac{E_\mu(n)}{n^\beta}\ge
      \frac{A(\beta)}{[\zeta(\beta^{-1})]^\beta}
      ,
    \end{align}
    where 
    \begin{align}
      \label{Abeta}
      A(\beta):=
      \sup_{\delta\in(1/2,1)}
      (1-\eta(\delta))^\beta
      \int_{\sqrt{\delta}}^1 \frac{(1-u)^2\, du}{u\, (-\ln u)^{\beta+1}}
    \end{align}
    and $\eta(\delta)$ is the entropy of binary distribution
    $(\delta,1-\delta)$,
    \begin{align*}
      \eta(\delta):= -\delta\log \delta-(1-\delta)\log (1-\delta)
      .
    \end{align*}
  \item If $\lim_{k\rightarrow\infty} p_k/P(K_i=k)=0$ then $E_\mu(n)$
    obeys (\ref{limEnUDP}).
  \end{enumerate}
\end{proposition}

Now let us introduce a~similar ergodic process over a~finite alphabet.
For this goal we use a~transformation of processes over an infinite
alphabet into processes over a~finite alphabet that preserves
stationarity and (non)ergodicity and does not distort entropy too
much, as we have shown in \cite{Debowski10}. We call this
transformation stationary (variable-length) coding. (The same or
a~similar construction has been considered in
\cite{CariolaroPierobon77,GrayKieffer80,TimoBlackmoreHanlen2007}.) It
is a~composition of two operations.

First, let a~function $f:\mathbb{X}\rightarrow\mathbb{Y}^*$, called
a~coding function, map symbols from alphabet $\mathbb{X}$ into strings
over another alphabet $\mathbb{Y}$. We define its extension to double
infinite sequences $f^{\mathbb{Z}}:\mathbb{X}^{\mathbb{Z}}\rightarrow
\mathbb{Y}^{\mathbb{Z}}\cup(\mathbb{Y}^*\times\mathbb{Y}^*)$ as
\begin{align}
  \label{InfExtension}
  f^{\mathbb{Z}}((x_i)_{i\in\mathbb{Z}})&:=
  ... f(x_{-1})f(x_{0})\textbf{.}f(x_1)f(x_2)...
  ,
\end{align}
where $x_i\in\mathbb{X}$ and the bold-face dot separates the $0$-th
and the first symbol. Then for a~stationary process
$(X_i)_{i\in\mathbb{Z}}$ on $(\Omega,\mathcal{J},P)$, where variables
$X_i$ take values in space $(\mathbb{X},\mathcal{X})$, we introduce
process
\begin{align}
  \label{DefY}
  (Y_i)_{i\in\mathbb{Z}}:=f^{\mathbb{Z}}((X_i)_{i\in\mathbb{Z}})
  ,
\end{align}
where variables $Y_i$ take values in space $(\mathbb{Y},\mathcal{Y})$,
as long as the right hand side is a~double infinite sequence almost
surely.

The second operation is as follows.  Transformation (\ref{DefY}) does
not preserve stationarity in general but process
$(Y_i)_{i\in\mathbb{Z}}$ is asymptotically mean stationary (AMS) under
mild conditions \cite[Proposition 2.3]{Debowski10}, which are
satisfied in the setting considered further. Then for the distribution
\begin{align}
  P((Y_i)_{i\in\mathbb{Z}}\in\cdot)=\nu
\end{align}
and the shift operation
$T((y_i)_{i\in\mathbb{Z}}):=(y_{i+1})_{i\in\mathbb{Z}}$ there exists
a~stationary measure
\begin{align}  
  \label{BarMu}
  \bar\nu(A):=\lim_{n\rightarrow\infty} \frac{1}{n}  \sum_{i=0}^{n-1}
  \nu\circ T^{-i}(A)
  ,
\end{align}
called the stationary mean of $\nu$ \cite{GrayKieffer80,Debowski10}.
It is convenient to suppose that probability space
$(\Omega,\mathcal{J},P)$ is rich enough to support a~process $(\bar
Y_i)_{i\in\mathbb{Z}}$ with the distribution
\begin{align}
  P((\bar Y_i)_{i\in\mathbb{Z}}\in\cdot)=\bar\nu
  .
\end{align}
Whereas process $(Y_i)_{i\in\mathbb{Z}}$ need not be stationary,
process $(\bar Y_i)_{i\in\mathbb{Z}}$ is stationary and will be called
the stationary (variable-length) coding of $(X_i)_{i\in\mathbb{Z}}$. 

Processes $(X_i)_{i\in\mathbb{Z}}$, $(Y_i)_{i\in\mathbb{Z}}$, and
$(\bar Y_i)_{i\in\mathbb{Z}}$ have isomorphic shift-invariant algebras
for some nice coding functions, called synchronizable injections
\cite[Proposition 3.3]{Debowski10}. For example, for the infinite
alphabet $\mathbb{X}=\mathbb{N}\times\klam{0,1}$, let us assume the
ternary alphabet $\mathbb{Y}=\klam{0,1,2}$ and the coding function
\begin{align}
  \label{ConjCode}
  f(k,z)=b(k)z2
  ,
\end{align}
where $b(k)\in\klam{0,1}^+$ is the binary representation of a~natural
number $k$ stripped of the leading digit $1$.  Coding function
(\ref{ConjCode}) is an instance of a~synchronizable injection. Hence
we have the following fact:
\begin{proposition}
  \label{theoEnCoded}
  Let $(\bar Y_i)_{i\in\mathbb{Z}}$ be the stationary coding obtained
  from applying the coding function (\ref{ConjCode}) to the generalized
  Santa Fe process (\ref{exMixing}). Process $(\bar
  Y_i)_{i\in\mathbb{Z}}$ is nonergodic if $p_k=0$ and ergodic if
  $p_k\in(0,1)$.
\end{proposition}

Notice, however, that the stationary coding of a~mixing process is not
mixing for a~synchronizable coding function in general. For example,
if we take the generalized Santa Fe process and the coding function
$f(k,z)=01$, which is also a~synchronizable injection, the stationary
coding $(\bar Y_i)_{i\in\mathbb{Z}}$ is not mixing because of periodic
oscillations in the realizations of the process
$(Y_i)_{i\in\mathbb{Z}}$. Such regular periods do not arise for the
generalized Santa Fe process and the coding function (\ref{ConjCode})
since variables $\abs{f(X_i)}$, where $\abs{w}$ is the length of
string $w$, differ from constants and are independent and identically
distributed. Thus, we conjecture that the resulted process $(\bar
Y_i)_{i\in\mathbb{Z}}$ is mixing for $p_k\in(0,1)$. 

Now let us consider block mutual information for the stationary coding
of the generalized Santa Fe process.  Let us recall that
$\bar\nu=P((\bar Y_i)_{i\in\mathbb{Z}}\in\cdot)$ and
$E_{\bar\nu}(m)=I\okra{\bar Y_{1:m};\bar Y_{m+1:2m}}$. As the last
new result, we will show this fact:
\begin{proposition}
  \label{theoEnMixingCoded}
  Let $(\bar Y_i)_{i\in\mathbb{Z}}$ be the stationary coding obtained
  from applying the coding function (\ref{ConjCode}) to the generalized
  Santa Fe process (\ref{exMixing}). Define the expansion rate
  $L:=\sred{\abs{f(X_i)}}$. The block mutual information
  $E_{\bar\nu}(m)$ for process $(\bar Y_i)_{i\in\mathbb{Z}}$ satisfies
  \begin{align}
    \label{limEnMixingCoded} 
    \limsup_{m\rightarrow\infty} \frac{E_{\bar\nu}(m)}{m^\beta}\le
    \frac{1}{L^\beta}
    \frac{(2-2^\beta)\Gamma(1-\beta)}{[\zeta(\beta^{-1})]^\beta}
    .
  \end{align}
  The lower limits in particular cases are as follows:
  \begin{enumerate}
  \item If $p_k\le P(K_i=k)$ then
    \begin{align}
      \label{limEnMixingCodedPower}
      \liminf_{m\rightarrow\infty} \frac{E_{\bar\nu}(m)}{m^\beta}\ge
      \frac{1}{L^\beta}
      \frac{A(\beta)}{[\zeta(\beta^{-1})]^\beta}
      .
    \end{align}
    where $A(\beta)$ is defined in (\ref{Abeta}).
  \item If $\lim_{k\rightarrow\infty} p_k/P(K_i=k)=0$ then 
  \begin{align}
    \label{limEnCodedUDP}
    \lim_{m\rightarrow\infty} \frac{E_{\bar\nu}(m)}{m^\beta}=
    \frac{1}{L^\beta}
    \frac{(2-2^\beta)\Gamma(1-\beta)}{[\zeta(\beta^{-1})]^\beta}
    .
  \end{align}
  \end{enumerate}
\end{proposition}
Proposition \ref{theoEnMixingCoded} follows from Proposition
\ref{theoEnMixing} by the conditional data processing inequality and
Chernoff bounds.  This proposition strengthens inequality
\begin{align}
  \label{limsupEBarn}
  \limsup_{m\rightarrow\infty} \frac{E_{\bar\nu}(m)}{m^\beta}>0
  ,
\end{align}
which follows for $p_k=0$ by \cite[Proposition 1.4]{Debowski10} and
\cite[Theorem 2]{Debowski11b}. 

The further organization of this paper is as follows.  The rate of
mutual information for the original and generalized Santa Fe processes
is discussed in Section \ref{secExcess}. The rate of mutual
information for the stationary coding is established in Section
\ref{secEncoding}.  Subsequently, the mixing property for the
generalized Santa Fe process is shown in Appendix \ref{secMixing}. As
an auxiliary result, we demonstrate that infinite direct products of
mixing
processes are also mixing. 

\section{The rate of mutual information}
\label{secExcess}

In this section we evaluate the rate of block mutual information for
the Santa Fe process and its mixing counterpart. The main tool is
conditional mutual information for stochastic processes as discussed,
e.g., in \cite{Pinsker60en,Debowski09}.


Here are some facts about conditional information that will be used,
cf., \cite{Debowski09}:
\begin{enumerate}
\item[(a)] continuity
  $I(X;(Y_k)_{k\in\mathbb{N}}|Z)=\lim_{n\rightarrow\infty}
  I(X;(Y_k)_{k=1}^n|Z)$,
\item[(b)] chain rule $I(X;Y,Z|W)= I(X;Y|W) +I(X;Z|Y,W)$, and
\item[(c)] equality $I(X;Y|Z)=0$ for $X$ and $Y$ conditionally
  independent given $Z$.
\end{enumerate}
Two simple corollaries of the chain rule will be used as well:
\begin{enumerate}
\item $I(X;Y)= I(X;Y;Z)+ I(X;Y|Z)$ for $H(X),H(Y)<\infty$,
  where we define triple information
\begin{align*}
  I(X;Y;Z) &
  :=I(X;Z)+I(Y;Z)
  -I((X,Y);Z)
  ,
\end{align*}
\item $I(X;Z|Y)=I(X;Z)$ for $X$ and $Y$ independent and
  conditionally independent given $Z$.
\end{enumerate}
The second identity follows from $I(X;(Y,Z))= I(X;Y)
+I(X;Z|Y)=I(X;Z) +I(X;Y|Z)$ where both $I(X;Y)=0$ and
$I(X;Y|Z)=0$.

Now we can evaluate block mutual information $E_\mu(n)$ for the Santa
Fe processes. The case of the original Santa Fe process is simpler and
will be considered separately to guide the reader through the more
complicated proof for the generalized process.

\begin{proof*}{Proposition \ref{theoEnUDP}}
  Notice that variables $Z_k$, $k\in\mathbb{N}$, are independent and
  conditionally independent given any finite block $X_{n:m}$. Hence
\begin{align*}
  I\okra{X_{1:n};(Z_k)_{k\in\mathbb{N}}}
  &= \sum_{k=1}^\infty I(X_{1:n};Z_k|Z_{1:k-1})
  = \sum_{k=1}^\infty I(X_{1:n};Z_k)  
  .
\end{align*}
Also $X_{1:n}$ and $X_{n+1:2n}$ are conditionally independent given
$(Z_{k})_{k\in\mathbb{N}}$. Hence
$I\okra{X_{1:n};X_{n+1:2n}|(Z_{k})_{k\in\mathbb{N}}}=0$.  Both results
yield
\begin{align}
  E_\mu(n)&=I\okra{X_{1:n};X_{n+1:2n}}
  \nonumber
  \\
  &= I\okra{X_{1:n};X_{n+1:2n};(Z_k)_{k\in\mathbb{N}}} +   
  I\okra{X_{1:n};X_{n+1:2n}|(Z_k)_{k\in\mathbb{N}}}
  \nonumber
  \\
  &=I\okra{X_{1:n};X_{n+1:2n};(Z_k)_{k\in\mathbb{N}}}
  \nonumber
  \\
  &=2I\okra{X_{1:n};(Z_k)_{k\in\mathbb{N}}}
  -I\okra{X_{1:2n};(Z_k)_{k\in\mathbb{N}}}
  \nonumber
  \\
  &=
  \sum_{k=1}^\infty [2I(X_{1:n};Z_k)-I(X_{1:2n};Z_k)]
  \nonumber
  \\
  &=
  \sum_{k=1}^\infty I(X_{1:n};X_{n+1:2n};Z_k)
  \label{EnTriple}
  .
\end{align}

Computing simple expressions
  \begin{align*}
  H(Z_k|X_{1:n})
  &
  = 1 \cdot P(\text{$K_i\not=k$ for all $i\in\klam{1,...,n}$})
  \\
  &
  \phantom{=} + 0 \cdot P(\text{$K_i=k$ for some
    $i\in\klam{1,...,n}$})
  ,
  \\
  I(X_{1:n};Z_k)
  &
  = P(\text{$K_i=k$ for some
    $i\in\klam{1,...,n}$})
  \\
  &
  =(1-[1-P(K_i=k)]^n)
  ,
\end{align*}
we obtain  the triple mutual information
\begin{align*}
  I(X_{1:n};X_{n+1:2n};Z_k)
  &
  =
  (1-[1-P(K_i=k)]^n)^2
\end{align*}
and the block mutual information
\begin{align}
  \label{EnUDP}
  E_\mu(n)=\sum_{k=1}^\infty
  \okra{1-\okra{1-\frac{A}{k^{1/\beta}}}^n}^2
  ,
\end{align}
where $A:=1/\zeta(\beta^{-1})$.

The right-hand side of (\ref{EnUDP}) equals up to an additive
constant $\le 1$ to the integral
\begin{align*}
\int_1^\infty \okra{1-\okra{1-\frac{A}{k^{1/\beta}}}^n}^2 dk
=\beta (An)^\beta \int_{(1-A)^n}^1 f_n(u) du
,
\end{align*}
where we use substitution 
\begin{align}
  \label{u}
  u:=\okra{1-Ak^{-1/\beta}}^n
\end{align}
and functions
\begin{align}
  \label{fnu}
  f_n(u):=\frac{(1-u)^2}{u^{1-1/n}[n(1-u^{1/n})]^{\beta+1}}
  .
\end{align}

We have the limit
\begin{align*}
  \lim_{n\rightarrow\infty} f_n(u) =
  f(u):=\frac{(1-u)^2}{u(-\ln u)^{\beta+1}}
\end{align*}
with the upper bound
\begin{align*}
  \frac{f_n(u)}{f(u)}&\le\sup_{0< x< 1}
  \frac{x(-\ln x)^{\beta+1}}{(1-x)^{\beta+1}}= 1
  ,
  &
  u,\beta&\in(0,1)
  .
\end{align*}
Moreover, function $f(u)$ is integrable on $u\in(0,1)$.  Hence
\begin{align*}
  \lim_{n\rightarrow\infty} \frac{E_\mu(n)}{n^\beta}=
  \beta A^\beta
  \int_0^1 f(u) du
\end{align*}
follows by the dominated convergence theorem.  

It remains to compute $\int f(u) du$.  Putting $t:=-\ln u$ yields
\begin{align*}
  \int_0^1  f(u) du
  &
  =
  \int_0^\infty (1-e^{-t})^2\,t^{-\beta-1} dt
  \\
  &
  =
  \int_0^\infty [e^{-2t}-1-2(e^{-t}-1)]\,t^{-\beta-1} dt
  \\
  &
  =(2-2^\beta)\beta^{-1}\Gamma(1-\beta)
  ,
\end{align*}
where integral 
\begin{align*}
  &\int_0^\infty (e^{-kt} -1)t^{-\beta-1}dt
  =(e^{-kt}-1)(-\beta^{-1})t^{-\beta}|_0^\infty 
  \\
  &\qquad
  -
  \int_0^\infty (-ke^{-kt})(-\beta^{-1})t^{-\beta}dt
  =-k^\beta\beta^{-1}\Gamma(1-\beta)
\end{align*}
can be integrated by parts for the considered $\beta$.
\end{proof*}

Next, we prove the more general statement, partly using the preceding
proof.

\begin{proof*}{Proposition \ref{theoEnMixing}}
  Observe that processes $\tilde Z_k:=(Z_{ik})_{i\in\mathbb{Z}}$,
  where $k\in\mathbb{N}$, are independent and conditionally
  independent given any finite block $X_{n:m}$. Also $X_{1:n}$ and
  $X_{n+1:2n}$ are conditionally independent given
  $(\tilde Z_{k})_{k\in\mathbb{N}}$. Thus we obtain
  \begin{align*}
    E_\mu(n)&=
    \sum_{k=1}^\infty I(X_{1:n};X_{n+1:2n};\tilde Z_k)
  \end{align*}
  by replacing $Z_k$ with $\tilde Z_k$ in derivation (\ref{EnTriple})
  from the previous proof.

  By the assumed Markov property, process $\tilde
  Z_k=(Z_{ik})_{i\in\mathbb{Z}}$ is independent from $X_{1:n}$ given
  $(Z_{ik})_{1\le i\le n}$. This yields
  \begin{align*}
    I\okra{X_{1:n};X_{n+1:2n};\tilde Z_k}
    =
    2I\okra{X_{1:n};(Z_{ik})_{1\le i\le n}}
    -
    I\okra{X_{1:2n};(Z_{ik})_{1\le i\le 2n}}
    .
  \end{align*}
  The expressions on the right-hand side can be analyzed as
  \begin{align*}
    I\okra{X_{1:n};(Z_{ik})_{1\le i\le n}}
    =
    \sum_{i=1}^n I\okra{X_{i};Z_{ik}|X_{1:i-1}}
  \end{align*}
  because $(Z_{ik})_{1\le i\le n}$ is independent from $X_i$ given
  $Z_{ik}$ and $X_{1:i-1}$. Moreover,
  \begin{align*}
    I\okra{X_{i};Z_{ik}|X_{1:i-1}}
    &=
    H\okra{Z_{ik}|X_{1:i-1}}-H\okra{Z_{ik}|X_{1:i}}
    .
  \end{align*}

  To evaluate the conditional entropies, put
  $a_{nk}:=\eta(P(Z_{ik}=z|Z_{i-n,k}=z))$ and $b_k:=P(K_i=k)$.  Notice
  that by the Markovity of $(Z_{ik})_{i\in\mathbb{Z}}$ we have
  \begin{align*}
    H\okra{Z_{ik}|X_{1:i-1}} 
    =&\sum_{n=1}^{i-1} a_{nk} P(K_j\neq k\text{
      for } i-n< j\le i-1)P(K_{i-n}=k)
    \nonumber
    \\
    &+\eta(P(Z_{ik}=z)) P(K_j\neq k\text{ for } 1\le j\le i-1)
    \nonumber
    \\
    =&\sum_{n=1}^{i-1} a_{nk} b_k(1-b_k)^{n-1} + (1-b_k)^{i-1} 
    .
  \end{align*}
  Similarly, since $a_{0k}=0$, we obtain
  \begin{align*}
    H\okra{Z_{ik}|X_{1:i}} 
    =&\sum_{n=0}^{i-1} a_{nk} P(K_j\neq k\text{
      for } i-n< j\le i)P(K_{i-n}=k)
    \nonumber
    \\
    &+\eta(P(Z_{ik}=z)) P(K_j\neq k\text{ for } 1\le j\le i)
    \nonumber
    \\
    =&\sum_{n=1}^{i-1} a_{nk} b_k(1-b_k)^{n}+ (1-b_k)^{i}
    .
  \end{align*}

Thus we may reconstruct
\begin{align*}
  I\okra{X_{i};Z_{ik}|X_{1:i-1}}
  =&\sum_{n=1}^{i-1} a_{nk} b_k^2(1-b_k)^{n-1}
  + \kwad{(1-b_k)^{i-1}-(1-b_k)^{i}}
  ,
  \\
  I\okra{X_{1:n};(Z_{ik})_{1\le i\le n}}
  =&\sum_{m=1}^{n-1} (n-m) a_{mk} b_k^2(1-b_k)^{m-1}
  + \kwad{1-(1-b_k)^{n}}
  ,
\end{align*}
and
\begin{align*}
  I\okra{X_{1:n};X_{n+1:2n};\tilde Z_k}
  =&-\sum_{m=1}^{n-1} m a_{mk} b_k^2(1-b_k)^{m-1}
  \\
  &-\sum_{m=n}^{2n-1} (2n-m) a_{mk} b_k^2(1-b_k)^{m-1}
  +\kwad{1-(1-b_k)^{n}}^2
  .
\end{align*}

For a~fixed $b_k$, we see that $I\okra{X_{1:n};X_{n+1:2n};\tilde
  Z_k}$ is minimized for $a_{mk}=1$.  This case arises when $p_k=1/2$
and $(Z_{ik})_{i\in\mathbb{Z}}$ are IID. A~direct evaluation yields
then $H\okra{Z_{ik}|X_{1:i-1}}=1$,
$H\okra{Z_{ik}|X_{1:i}}=(1-b_k)$, $I\okra{X_{1:n};(Z_{ik})_{1\le
    i\le n}}=nb_k$, and $I\okra{X_{1:n};X_{n+1:2n};\tilde
  Z_k}=0$. In this way we have proved that
\begin{align}
  \label{BkEquality}
  \sum_{m=1}^{n-1} m b_k^2(1-b_k)^{m-1}
  +
  \sum_{m=n}^{2n-1} (2n-m) b_k^2(1-b_k)^{m-1}
  =
  \kwad{1-(1-b_k)^{n}}^2
  .
\end{align}
On the other hand, $I\okra{X_{1:n};X_{n+1:2n};\tilde Z_k}$ is maximized for
$a_{mk}=0$. This holds if $p_k=0$ or $p_k=1$. For $p_k=0$, the process
$(X_{i})_{i\in\mathbb{Z}}$ collapses to (\ref{exUDP}).

By equality (\ref{BkEquality}), we obtain
\begin{align*}
 I\okra{X_{1:n};X_{n+1:2n};\tilde Z_k}\in
 \kwad{
   (1 - \epsilon)\kwad{1-(1-b_k)^{n}}^2
   , 
   \kwad{1-(1-b_k)^{n}}^2
 }
\end{align*}
if $a_{mk}\le \epsilon$ for $m\le 2n-1$. To bound coefficients
$a_{mk}$, observe
\begin{align*}
  P(Z_{ik}=z|Z_{i-n,k}=z)\ge (1-p_k)^n
  .
\end{align*}
Hence $a_{mk}\le\eta(\delta)$ for $m\le 2n-1$ if
$(1-p_k)^{2n}\ge\delta\ge 1/2$. Thus we obtain
\begin{align}
  \label{EnBoundsGeneral}
  E_\mu(n)
  \in
  \kwad{
    (1-\eta(\delta))
    \sum_{k\in\mathbb{N}:(1-p_k)^{2n}\ge\delta} \kwad{1-(1-b_k)^{n}}^2
    ,    
    \sum_{k\in\mathbb{N}} \kwad{1-(1-b_k)^{n}}^2
  }
  .
\end{align}
The most tedious part of the proof is completed.

The limiting behavior of the upper bound in (\ref{EnBoundsGeneral})
has been analyzed in the proof of Proposition \ref{theoEnUDP}, and by
that reasoning (\ref{limEnMixing}) holds. Now we will consider the
limit of the lower bound in (\ref{EnBoundsGeneral}).  As in the
previous proof, we will approximate the respective sum with an
integral.  Recall that $b_k=Ak^{-1/\beta}$ with
$A=1/\zeta(\beta^{-1})$.  Let us define $b_k$ for real $k$ in the same
way.
\begin{enumerate}
\item For $p_k\le b_k$: notice that $(1-b_k)^{2n}\ge\delta$ implies
  $(1-p_k)^{2n}\ge\delta$. Thus 
  $E_\mu(n)/(1 - \eta(\delta))+1$ is greater than
  \begin{align*}
    \int_{(1-b_k)^{n}\ge\sqrt{\delta}}^\infty 
    \okra{1-\okra{1-b_k}^n}^2 dk
    =
    \beta (An)^\beta \int_{\sqrt{\delta}}^1 f_n(u) du
    ,
  \end{align*}
  where we use substitution (\ref{u}) and functions (\ref{fnu}). This
  yields (\ref{limEnMixingPower}) by the dominated convergence theorem.
\item For $\lim_{k} p_k/b_k=0$: let $k(n)$ be the largest number $k$
  such that $(1-p_k)^{2n}<\delta$ or put $k(n)=1$ if there is no such
  number. Then $E_\mu(n)/(1 - \eta(\delta))+1$ is greater than
  \begin{align*}
    \int_{k(n)}^\infty 
    \okra{1-\okra{1-b_k}^n}^2 dk
    =
    \beta (An)^\beta \int_{u(n)}^1 f_n(u) du
    ,
  \end{align*}
  where $u(n):=\okra{1-b_{k(n)}}^n$. We have $\lim_{n} u(n)=0$ if
  $\lim_{n} k(n)<\infty$. On the other hand, if $\lim_{n} k(n)=\infty$
  then we use $\liminf_{n} np_{k(n)}>-\ln\sqrt{\delta}$ and $\lim_{k}
  p_k/b_k=0$ to infer $\lim_{n} nb_{k(n)}=\infty$ and hence $\lim_{n}
  u(n)=0$. Thus the dominated convergence theorem in both cases yields
  \begin{align*}
    \liminf_{n\rightarrow\infty} \frac{E_\mu(n)}{n^\beta}
    \ge (1-\eta(\delta)) \beta A^\beta \int_0^1 f(u)du
    .
  \end{align*}
  Taking  $\delta\rightarrow 1$ gives (\ref{limEnUDP}). 
\end{enumerate}
\end{proof*}

\section{Encoding into a~finite alphabet}
\label{secEncoding}

In this section we study the rate of mutual information for the
stationary coding of the generalized Santa Fe process.  Let $\abs{w}$
be the length of string $w$ and let $(X_i)_{i\in\mathbb{Z}}$ denote
the generalized Santa Fe process.  For the coding function
(\ref{ConjCode}), regardless of the value of $p_k$, the expansion rate
\begin{align}
  \lim_{n\rightarrow\infty} \frac{1}{n}\sum_{i=1}^{n}\abs{f(X_i)}
\end{align}
is almost surely constant and equals the expansion rate
$L:=\sred{\abs{f(X_i)}}$. Hence the stationary coding $(\bar
Y_i)_{i\in\mathbb{Z}}$ can be constructed as detailed below. This
construction was formally introduced in \cite[Section 6]{Debowski10}
and justified by \cite[Proposition 2.3]{Debowski10}.

Suppose that probability space $(\Omega,\mathfrak{J},P)$ is
sufficiently rich to support some previously unmentioned random
variable $N:\Omega\rightarrow\mathbb{N}\cup\klam{0}$, called a~random
shift, and a~nonstationary process $(\bar X_i)_{i\in\mathbb{Z}}$ where
$\bar X_i:\Omega\rightarrow\mathbb{X}$. We assume that $N$ and $(\bar
X_i)_{i\in\mathbb{Z}}$ are conditionally independent given $\bar X_0$
and their distribution is
\begin{align}
  \label{DefBarX}
  P(\bar X_{k:l}=x_{k:l})
  &= 
  P(X_{k:l}=x_{k:l})\cdot \frac{\abs{f(x_0)}}{L} 
  ,
  &
  k&\le 0\le l,
  \\
  \label{DefN}
  P(N=n|\bar X_{0}=x_{0})
  &=
  \frac{
    \boole{0\le n\le \abs{f(x_0)}-1}
  }{
    \abs{f(x_0)}
  }
  ,
  &
  n&\in\mathbb{N}\cup\klam{0}
  .
\end{align}
Process $(\bar Y_i)_{i\in\mathbb{Z}}$ with the desired distribution
$\bar\nu=P((\bar Y_i)_{i\in\mathbb{Z}}\in\cdot)$, where
$\nu=P((Y_i)_{i\in\mathbb{Z}}\in\cdot)$ for
$(Y_i)_{i\in\mathbb{Z}}=f^{\mathbb{Z}}((X_i)_{i\in\mathbb{Z}})$,
can be obtained as
\begin{align}
  \label{DefBarY}
  (\bar Y_i)_{i\in\mathbb{Z}}=T^{-N} f^{\mathbb{Z}}((\bar X_i)_{i\in\mathbb{Z}})
  ,
\end{align}
where $T((y_i)_{i\in\mathbb{Z}}):=(y_{i+1})_{i\in\mathbb{Z}}$ is the
shift operation.  

\begin{lemma}
  \label{theoBarXX}
  Denote blocks $\bar X_{k:l}$ with $\bar X_0$ removed as $\bar
  X_{k:l\setminus 0}$. For the Santa Fe processes variables $\bar
  X_{k:l\setminus 0}$ and $X_{k:l\setminus 0}$ have the same
  distribution.
\end{lemma}
\begin{proof}
Notice that $\abs{f(X_0)}$ does not depend on $Z_{0,K_0}$ and $K_0$ is
independent of $X_{k:l\setminus 0}$. Hence
\begin{align}
  P(\bar X_{k:l\setminus 0})
  &=
  P(X_{k:l\setminus 0})
  \sum_{x_0} \frac{\abs{f(x_0)}}{L} 
  P(X_0=x_0|X_{k:l\setminus 0})    
  \nonumber\\
  &=
  P(X_{k:l\setminus 0})
  \sum_{k_0,z_0} \frac{\abs{f(k_0,1)}}{L}
  P(K_0=k_0)
  P(Z_{0,k_0}=z_0|X_{k:l\setminus 0})
  \nonumber\\
  \label{BarXX}
  &=
  P(X_{k:l\setminus 0})
  .
\end{align}  
\end{proof}

In the following we write $L_i:=\abs{f(X_i)}$ and $\bar
L_i:=\abs{f(\bar X_i)}$. Variables $L_i$ are independent and
identically distributed. For these variables we define indices
\begin{align}
  L^+_t&:=\frac{1}{t}\log \sred{2^{tL_i}}
  ,
  \\
  L^-_t&:=-\frac{1}{t}\log \sred{2^{-tL_i}}
  ,
\end{align}
where $t>0$.  For the given distribution of $L_i$, we have
$0<L^-_t,L^+_t<\infty$ for sufficiently small $t$. By the Jensen
inequality $L^+_t$ is a~growing function of $t$ and $L^-_t$ is
a~decreasing function of $t$. Jensen inequality implies also
$L^-_t\le L\le L^+_t$.
\begin{lemma}
  We have
  \begin{align}
    \label{limL}
    \lim_{t\rightarrow 0}L^+_t=\lim_{t\rightarrow 0}L^-_t=L
    .
  \end{align}
\end{lemma}
\begin{proof}
  Consider function $g(t,x)=t^{-2}(2^{tx}-1-tx)$. For $x>0$, it is
  a~growing function of $t$. Consider next such a~$t_0$ that
  $\sred{2^{t_0L_i}}<\infty$.  For $0<t\le t_0$, we obtain
  \begin{align*}
   \sred{2^{tL_i}}
   = 1+t\sred{L_i}+t^2\sred{g(t,L_i)}
   \le 1+t\sred{L_i}+t^2\sred{g(t_0,L_i)}
   .
  \end{align*}
  This yields
  \begin{align*}
    L\le L^+_t\le \frac{1}{t}\log
    \okra{1+t\sred{L_i}+t^2\sred{g(t_0,L_i)}}
    \xrightarrow{t\rightarrow 0} L
    .
  \end{align*}
  On the other hand, for $t>0$, we have
  \begin{align*}
   \sred{2^{-tL_i}}
   \le 1-t\sred{L_i}+\frac{t^2\sred{L_i^2}}{2}
  \end{align*}
  Hence
  \begin{align*}
    L\ge L^-_t\ge -\frac{1}{t}\log
    \okra{1-t\sred{L_i}+\frac{t^2\sred{L_i^2}}{2}}
    \xrightarrow{t\rightarrow 0} L
    .
  \end{align*} 
\end{proof}

Define events
  \begin{align}
    S^+_n
    &:=    
    \okra{\sum_{i=1}^{n} L_i < n(L^+_t+\epsilon)}
    ,
    \\
    S^-_n
    &:=    
    \okra{\sum_{i=1}^{n} L_i > n(L^-_t-\epsilon)}
    ,
    \\
    T^+_n
    &:=
    \okra{\sum_{i=-n+1}^{-1} L_i < (n-1)(L^+_t+\epsilon)}
    ,
    \\
    T^-_n
    &:=
    \okra{\sum_{i=-n+1}^{-1} L_i > (n-1)(L^-_t-\epsilon)}
    .    
  \end{align}
Subsequently, we will use the Chernoff bounds:
\begin{lemma}
For $t>0$ and $\epsilon>0$,
\begin{align}
  \label{PSPlus}
  P\okra{{S^+_n}^c}
  &\le \frac{1}{2^{nt\epsilon}}
  ,
  \\
  \label{PSMinus}
  P\okra{{S^-_n}^c}
  &\le \frac{1}{2^{nt\epsilon}}
  ,
  \\
  \label{PTPlus}
  P\okra{{T^+_n}^c}
  &\le \frac{1}{2^{(n-1)t\epsilon}}
  ,
  \\
  \label{PTMinus}
  P\okra{{T^-_n}^c}
  &\le \frac{1}{2^{(n-1)t\epsilon}}
  .
\end{align}
\end{lemma}
\begin{proof}
  Because variables $L_i$ are independent and identically distributed,
  using Markov inequality we observe
  \begin{align*}
    P\okra{{S^+_n}^c}
    &=
    P\okra{2^{t\sum_{i} L_i} \ge 2^{tn(L^+_t + \epsilon)}}
    \le
    \frac{\sred{2^{t\sum_{i} L_i}}}{2^{tn(L^+_t + \epsilon)}}
    \le
    \frac{1}{2^{nt\epsilon}}
    ,
    \\
    P\okra{{S^-_n}^c}
    &=
    P\okra{2^{-t\sum_{i} L_i} \ge 2^{-tn(L^-_t - \epsilon)}}
    \le
    \frac{\sred{2^{-t\sum_{i} L_i}}}{2^{-tn(L^-_t - \epsilon)}}
    \le
    \frac{1}{2^{nt\epsilon}}
    .
  \end{align*}
  Analogously we obtain the claims for ${T^+_n}^c$ and ${T^-_n}^c$.
\end{proof}

Next, for an event $E$, we introduce conditional entropy $H(X|E)$ and
mutual information $I(X;Y|E)$ which are respectively the entropy of
variable $X$ and mutual information between variables $X$ and $Y$
taken with respect to probability measure $P(\cdot|E)$.
\begin{lemma}
  For the generalized Santa Fe process, let $\gamma:=\beta/(1-\beta)$
  and $s<\min(t/2,\gamma/2)$. Then for sufficiently large $n$,
  \begin{align}
    \label{PSPlusHSPlus}
    P\okra{{S^+_n}^c}
    H\okra{L_1\,\middle |\, {S^+_n}^c}
    &\le
    \frac{1}{2^{ns\epsilon}}
    ,
    \\
    \label{PTPlusHTPlus}
    P\okra{{T^+_n}^c}
    H\okra{L_{-1}\,\middle |\, {T^+_n}^c}
    &\le
    \frac{1}{2^{(n-1)s\epsilon}}
    .
  \end{align}
\end{lemma}
\begin{proof}
   We have
  \begin{align*}
    P(L_i=l)=\sum_{k=2^{l-1}}^{2^l-1}
    \frac{k^{-1/\beta}}{\zeta(\beta^{-1})}
    \le
    \int_{2^{l-2}}^{\infty}k^{-1/\beta}dk
    =\gamma 2^{-\gamma(l-2)}
    .
  \end{align*}
  Write $i(p):=-p\log p$. Then for sufficiently large $N$,
  \begin{align*}
    \sum_{l=N}^\infty &i(P(L_i=l))
    \\
    &
    \le
    \sum_{l=N}^\infty i(\gamma 2^{-\gamma(l-2)})
    \\
    &
    \le
    \sum_{l=N}^\infty \gamma 2^{2\gamma}2^{-\gamma l}
    \okra{-\log\gamma+\gamma l}
    \\
    &=
    \gamma 2^{2\gamma}
    \okra{
      \frac{2^{-\gamma N}}{1-2^{-\gamma}}\okra{-\log\gamma}
      +
      \frac{2^{-\gamma (N+1)}+N2^{-\gamma N}(1-2^{-\gamma})}{(1-2^{-\gamma})^2}
      \gamma
    }
    \\
    &\le
    A(\gamma) N2^{-\gamma N}
    .
  \end{align*}
  Let $M=n(L^+_t+\epsilon)$ and $N=\ceil{n\epsilon/2}-1$. Then,
  for sufficiently large $n$,
  \begin{align*}
    P\okra{{S^+_n}^c}
    &H\okra{L_1\,\middle |\, {S^+_n}^c}
    \\
    &=
    P\okra{{S^+_n}^c}
    \sum_{l=0}^\infty
    i\okra{
      \frac{
        P\okra{L_1=l,\sum_{i=2}^n L_i \ge M-l}
      }{
        P\okra{{S^+_n}^c}
      }
    }
    \\
    &=
    P\okra{{S^+_n}^c}
    \sum_{l=0}^\infty
    i\okra{
      \frac{
        P(L_1=l)P\okra{\sum_{i=2}^n L_i \ge M-l}
      }{
        P\okra{{S^+_n}^c}
      }
    }
    \\
    &\le    
    \sum_{l=0}^\infty
    i\okra{
      P(L_1=l)P\okra{\sum_{i=2}^n L_i \ge M-l}
    }
    \\
    &\le    
    \sum_{l=0}^{N-1}
    i\okra{
      P\okra{\sum_{i=2}^n L_i \ge M-l}
    }
    +
    \sum_{l=N}^\infty
    i\okra{
      P(L_1=l)
    }
    \\
    &\le
    N i\okra{
      P\okra{\sum_{i=2}^n L_i \ge M-N+1}
    }
    +
    \sum_{l=N}^\infty
    i\okra{
      P(L_1=l)
    }
    \\
    &\le
    \frac{n(n-1)t\epsilon^2}{2^{(n-1)t\epsilon/2}} 
    +
    \frac{n A(\gamma)\epsilon}{2^{(n-1)\gamma \epsilon/2}}
    \le
    \frac{1}{2^{ns\epsilon}}
    .
  \end{align*}
   Analogously we obtain the claim for ${T^+_n}^c$.
\end{proof}

Now, define events
\begin{align}
  \bar S^+_n
  &:=    
  \okra{\sum_{i=1}^{n} \bar L_i < n(L^+_t+\epsilon)}
  ,
  \\
  \bar S^-_n
  &:=    
  \okra{\sum_{i=1}^{n} \bar L_i > n(L^-_t-\epsilon)}
  ,
  \\
  \bar T^+_n
  &:=
  \okra{\sum_{i=-n+1}^{-1} \bar L_i < (n-1)(L^+_t+\epsilon)}
  ,
  \\
  \bar T^-_n
  &:=
  \okra{\sum_{i=-n+1}^{-1} \bar L_i > (n-1)(L^-_t-\epsilon)}
  ,
  \\
  \bar B
  &:=\okra{\bar L_0\le l}
  .
\end{align}
\begin{lemma}
  For $m\ge n(L^+_t+\epsilon)+l$ we have
  \begin{align}
    I\okra{\bar X_{-n+1:-1};\bar X_{1:n}\middle|
      \bar S^+_n\cap \bar T^+_n}
    &=
    I\okra{\bar X_{-n+1:-1};\bar X_{1:n}\middle|
      \bar B\cap \bar S^+_n\cap \bar T^+_n}
    \nonumber
    \\
    \label{BarXCBarYC}
    &\le
    I\okra{\bar Y_{-m+1:0};\bar Y_{1:m}\middle|
      \bar B\cap \bar S^+_n\cap \bar T^+_n}
    ,
  \end{align}
  whereas for $m\le (n-1)(L^-_t-\epsilon)$ we have
  \begin{align}
    \label{BarYCBarXC}
    I\okra{\bar Y_{-m+1:0};\bar Y_{1:m}\middle|
      \bar S^-_n\cap \bar T^-_n}
    &\le
    I\okra{\bar X_{-n+1:0};\bar X_{0:n}\middle|
      \bar S^-_n\cap \bar T^-_n}
    .
  \end{align}
\end{lemma}
\begin{proof}
  The claims follow by equality (\ref{DefBarY}) and conditional data
  processing inequality
  \begin{align*}
    I(U';V'|C)\le I(U;V|C)
    ,
  \end{align*}
  which holds if equalities $U'=g(U)$ and $V'=h(V)$ are satisfied on
  $C$.
\end{proof}

There is an additional fact that we shall use.  
Let $I_C$ be the indicator function of event $C$. Observe that
\begin{align}
  P(C)I(X;Y|C) 
  &\le P(C)I(X;Y|C)+P(C^c)I(X;Y|C^c)
  \nonumber
  \\
  &=I(X;Y|I_C)= I(X;Y)-I(X;Y;I_C)
  \label{PCIC}
  ,
\end{align}
where $\abs{I(X;Y;I_C)}\le H(I_C)\le 1$ by the information diagram
\cite{Yeung02}. 

\begin{proof*}{Proposition \ref{theoEnMixingCoded}}
  Observe that
  \begin{align}
    H\okra{X_1\middle|{S^+_n}^c}&
    = H\okra{X_1|L_1} + 
    H\okra{L_1\middle|{S^+_n}^c}
    \nonumber\\
    \label{HXSPlus}
    &\le H\okra{X_0} + 
    H\okra{L_1\middle|{S^+_n}^c}
    ,
    \\
    H\okra{X_{-1}\middle|{T^+_n}^c}&
    = H\okra{X_{-1}|L_{-1}} + 
    H\okra{L_{-1}\middle|{T^+_n}^c}
    \nonumber\\
    \label{HXTPlus}
    &\le H\okra{X_0} + 
    H\okra{L_{-1}\middle|{T^+_n}^c}
    .
  \end{align}
  because $X_1$ is conditionally independent from ${S^+_n}^c$ given
  $L_1$ and $X_{-1}$ is conditionally independent from ${T^+_n}^c$
  given $L_{-1}$.  Now, assume that $n$ is sufficiently large so that
  bounds (\ref{PSPlusHSPlus}) and (\ref{PTPlusHTPlus}) hold true.  For
  brevity, define events
  \begin{align*}
    C^+_n
    &:=
    T^+_n \cap S^+_n
    ,
    \\
    \bar C^+_n
    &:=
    \bar T^+_n \cap \bar S^+_n
    .
  \end{align*}
  Then inequalities (\ref{HXSPlus}), (\ref{HXTPlus}), (\ref{PSPlus}),
  (\ref{PTPlus}), (\ref{PSPlusHSPlus}), and (\ref{PTPlusHTPlus}) yield
  \begin{align}
    P\okra{{C^+_n}^c}&I\okra{X_{-n+1:-1};X_{1:n}\middle|{C^+_n}^c}
    \nonumber\\
    &\le P\okra{{S^+_n}^c}nH\okra{X_1\middle|{S^+_n}^c}
    + P\okra{{T^+_n}^c}(n-1)H\okra{X_{-1}\middle|{T^+_n}^c}
    \nonumber\\
    \label{PCIXC}    
    &\le \frac{2n}{2^{(n-1)t\epsilon}}H\okra{X_0}
    + \frac{2n}{2^{(n-1)s\epsilon}}
    .
  \end{align}
  Moreover, assume that $m\ge n(L^+_t+\epsilon)+l$. Then applying
  subsequently (\ref{PCIC}), (\ref{BarXCBarYC}),
  (\ref{BarXX}), (\ref{PCIC}), and (\ref{PCIXC}) we obtain
  \begin{align}
    E_{\bar\nu}(m)
    &=
    I\okra{\bar Y_{-m+1:0};\bar Y_{1:m}}
    \nonumber\\
    &\ge 
    P\okra{\bar B\cap \bar C^+_n}
    I\okra{\bar Y_{-m+1:0};\bar Y_{1:m}\middle|\bar B\cap \bar C^+_n}
    - 1
    \nonumber\\
    &\ge 
    P\okra{\bar B}P\okra{\bar C^+_n}
    I\okra{\bar X_{-n+1:-1};\bar X_{1:n}\middle|\bar C^+_n}
    - 1
    \nonumber\\
    &= 
    P\okra{\bar B}P\okra{C^+_n}
    I\okra{X_{-n+1:-1};X_{1:n}\middle|C^+_n}
    - 1
    \nonumber\\
    &\ge 
    P\okra{\bar B}P\okra{C^+_n}
    \kwad{
     I\okra{X_{-n+1:0};X_{1:n}\middle|C^+_n}-H(X_0|C^+_n)
    }- 1
    \nonumber\\
    &\ge 
    P\okra{\bar B}P\okra{C^+_n}
    I\okra{X_{-n+1:0};X_{1:n}\middle|C^+_n}
    -H(X_0) - 1
    \nonumber\\
    &
    \ge
    P\okra{\bar B}
    \kwad{
      E_\mu(n)-1
      -P\okra{{C^+_n}^c}I\okra{X_{-n+1:0};X_{1:n}\middle|{C^+_n}^c}
    }-H(X_0) - 1    
    \nonumber\\
    &
    \ge
    P\okra{\bar B}
    E_\mu(n)
    -P\okra{{C^+_n}^c}I\okra{X_{-n+1:-1};X_{1:n}\middle|{C^+_n}^c}
    -2H(X_0) - 2    
    \nonumber\\
    \label{EnuEmu}
    &
    \ge
    P\okra{\bar B}E_\mu(n)
    -\kwad{\frac{2n}{2^{(n-1)t\epsilon}}+2}H\okra{X_0}
    -\frac{2n}{2^{(n-1)s\epsilon}} - 2
    .
  \end{align}

  Next, define events
  \begin{align*}
    C^-_n
    &:=
    T^-_n \cap S^-_n
    ,
    \\
    \bar C^-_n
    &:=
    \bar T^-_n \cap \bar S^-_n
    .    
  \end{align*}
  By (\ref{PSMinus}) and (\ref{PTMinus}) we have
  \begin{align}
    P\okra{{C^-_n}^c} &I\okra{\bar Y_{-m+1:0};\bar
      Y_{1:m}\middle|{C^-_n}^c}
    \nonumber\\
    \label{PCIYC}
    &\le \okra{P\okra{{S^-_n}^c}+P\okra{{T^-_n}^c}} m\log 3
    \le \frac{2m}{2^{(n-1)t\epsilon}}\log 3
    .
  \end{align}
  Assume that $m\le (n-1)(L^-_t-\epsilon)$. Then applying subsequently
  (\ref{BarXX}), (\ref{PCIC}), (\ref{BarYCBarXC}), (\ref{PCIC}), and
  (\ref{PCIYC}) we obtain
  \begin{align}
    E_{\mu}(n)
    &=
    I\okra{X_{-n+1:0};X_{1:n}}
    \nonumber\\
    &\ge
    I\okra{X_{-n+1:-1};X_{1:n}} 
    \nonumber\\
    &=
    I\okra{\bar X_{-n+1:-1};\bar X_{1:n}} 
    \nonumber\\
    &\ge
    I\okra{\bar X_{-n+1:0};\bar X_{0:n}} - 2H(\bar X_0) 
    \nonumber\\
    &\ge
    P\okra{\bar C^-_n}
    I\okra{\bar X_{-n+1:0};\bar X_{0:n}\middle|\bar C^-_n} 
    -1 - 2H(\bar X_0) 
    \nonumber\\
    &\ge
    P\okra{\bar C^-_n}
    I\okra{\bar Y_{-m+1:0};\bar Y_{1:m}\middle|\bar C^-_n} 
    - 2H(\bar X_0) -1
    \nonumber\\
    &
    \ge
    E_{\bar\nu}(m)-1
    -P\okra{{C^-_n}^c}I\okra{\bar Y_{-m+1:0};\bar Y_{1:m}\middle|{C^-_n}^c}
    -2H(\bar X_0) - 1
    \nonumber\\
    \label{EmuEnu}
    &
    \ge
    E_{\bar\nu}(m)
    -\frac{2m}{2^{(n-1)t\epsilon}}\log 3
    -2H(\bar X_0) - 2
    .
  \end{align}

  From bounds (\ref{EnuEmu}) and (\ref{EmuEnu}) we obtain
  \begin{align*}
    \frac{1}{[L^-_t-\epsilon]^\beta}
    \limsup_{n\rightarrow\infty} \frac{E_{\bar\nu}(n)}{n^\beta}
    &\ge
    \limsup_{m\rightarrow\infty} \frac{E_{\bar\nu}(m)}{m^\beta}
    \ge
    \frac{P(\bar L_0\le l)}{[L^+_t+\epsilon]^\beta}
    \limsup_{n\rightarrow\infty} \frac{E_{\bar\nu}(n)}{n^\beta}
    ,
    \\
    \frac{1}{[L^-_t-\epsilon]^\beta}
    \liminf_{n\rightarrow\infty} \frac{E_{\bar\nu}(n)}{n^\beta}
    &\ge
    \liminf_{m\rightarrow\infty} \frac{E_{\bar\nu}(m)}{m^\beta}
    \ge
    \frac{P(\bar L_0\le l)}{[L^+_t+\epsilon]^\beta}
    \liminf_{n\rightarrow\infty} \frac{E_{\bar\nu}(n)}{n^\beta}
    .
  \end{align*}
  If we consider $t\rightarrow 0$, $\epsilon\rightarrow 0$, and
  $l\rightarrow \infty$ then the requested claims will follow by equation
  (\ref{limL}) and Proposition \ref{theoEnMixing}.
\end{proof*}

\appendix
\renewcommand{\thesection}{\Alph{section}}

\section{Mixing properties}
\label{secMixing}

In this appendix we will discuss mixing properties of the generalized
Santa Fe process.  The setting makes use of
the $L^2$ space of complex valued functions.
  Then, for a~measure space $(\Omega,\mathcal{J},\mu)$ let
  \begin{align*}
    L_0^2(\Omega,\mathcal{J},\mu):=\klam{f\in
    L^2(\Omega,\mathcal{J},\mu): \int f d\mu=0}
  \end{align*}
  and denote the inner product $(f,g)_\mu:=\int f \bar g d\mu$ and
  the norm $\norm{f}_\mu:=\sqrt{(f,f)_\mu}$ for $f,g\in
  L^2(\Omega,\mathcal{J},\mu)$. Let also $T:\Omega\rightarrow\Omega$ be an
  invertible transformation that preserves the measure, $\mu\circ
  T^{-1}=\mu$.  The dynamical system $(\Omega,\mathcal{J},\mu,T)$ is called
    \emph{mixing} when
      $\displaystyle\lim_{n\rightarrow\infty} (f\circ T^n, g)_\mu=0
      \text{ for $f,g\in L_0^2(\Omega,\mathcal{J},\mu)$.}$
By the way, we know that any mixing dynamical
system 
is ergodic \cite[Chapter 1.\S6]{CornfeldFominSinai82en}.

The following proposition generalizes Theorem 2 from \cite[Chapter
10.\S 1]{CornfeldFominSinai82en}. Whereas the original claim deals
with finite direct products of dynamical systems, we will extend it
here to infinite products. To the best of our knowledge this
generalization has not been discussed in the literature so far.  The
proof is similar to the finite case, except for using a~different
orthonormal basis of the product space.
\begin{proposition}
  \label{theoStrongMixing}
  Let $(\Omega_j,\mathcal{J}_j,\mu_j,T_j)$, where $j\in\mathbb{N}$, be
  dynamical systems with probability measures $\mu_j(\Omega)=1$.
  Consider the direct product $(\Omega,\mathcal{J},\mu,T)$, where
  $\Omega=\times_{j=1}^\infty\Omega_j$,
  $\mathcal{J}=\otimes_{j=1}^\infty\mathcal{J}_j$,
  $\mu=\times_{j=1}^\infty\mu_j$, and
  $T(\omega)=(T_j(\omega_j))_{j\in\mathbb{N}}$ for
  $\omega=(\omega_j)_{j\in\mathbb{N}}$, $\omega_j\in\Omega_j$.  If
  $(\Omega_j,\mathcal{J}_j,\mu_j,T_j)$ are mixing then
  $(\Omega,\mathcal{J},\mu,T)$ is also mixing.
\end{proposition}
\begin{proof}
  Let $(e_{\alpha_j,j})_{\alpha_j\in A_j}$ be orthonormal bases of
  spaces $L^2(\Omega_j,\mathcal{J}_j,\mu_j)$ with $e_{0j}=1$ and
  $e_{\alpha_j,j}\in L_0^2(\Omega_j,\mathcal{J}_j,\mu_j)$. Then the
  set
  \begin{align}
    \label{Basis}
    \klam{e_\emptyset(\omega)=1}
    \cup
    \klam{e_\alpha(\omega)=\prod_{j=1}^{k}
    e_{\alpha_j,j}(\omega_j)}_{\alpha\in A_1\times A_2\times...\times
    (A_k\setminus\klam{0}),k=1,2,...}
  \end{align}
  with multi-indices $\alpha=(\alpha_1,\alpha_2,...,\alpha_k)$ is an
  orthonormal basis of the space $L^2(\Omega,\mathcal{J},\mu)$,
  cf., \cite[page 29]{Berezanskij86en}. (Orthogonality of set
  (\ref{Basis}) is obvious whereas its completeness follows from the
  completeness of the analogical orthonormal sets for finite products
  and the $L^2$-bounded martingale convergence.)  Let
  $\alpha,\alpha'\neq\emptyset$.  We have $e_\alpha,e_{\alpha'}\in
  L_0^2(\Omega,\mathcal{J},\mu)$ and
  \begin{align}
    \abs{(e_\alpha\circ T^n, e_\alpha')_\mu}
    &=
    \prod_{j=1}^{k} 
    \abs{(e_{\alpha_j,j}\circ T_j^n, e_{\alpha'_j,j})_{\mu_j}}
    \le     
    \abs{(e_{\alpha_k,k}\circ T_k^n, e_{\alpha'_k,k})_{\mu_k}}
    \label{ProductIneq}
  \end{align}
  by Schwarz inequality if $\alpha$ and $\alpha$ have the same length
  $k$.  Otherwise, $(e_\alpha\circ T^n, e_\alpha')_\mu=0$.  Hence
  $\lim_{n\rightarrow\infty} (e_\alpha\circ T^n, e_\alpha')_\mu=0$
  holds by the hypothesis.

  Any other functions $f,g\in L_0^2(\Omega,\mathcal{J},\mu)$ can be
  represented as series $f=\sum_{\alpha\neq\emptyset} f_\alpha
  e_\alpha$ and $g=\sum_{\alpha\neq\emptyset} g_\alpha e_\alpha$,
  where $\sum_{\alpha\neq\emptyset}
  \abs{f_\alpha}^2,\sum_{\alpha\neq\emptyset}
  \abs{g_\alpha}^2<\infty$.  Assume without loss of generality that
  $\norm{f}_\mu=\norm{g}_\mu=1$.  We will show that for every
  $\epsilon>0$, inequality $\abs{(f\circ T^n, g)_\mu}< \epsilon$ holds
  for sufficiently large $n$.  Let $F$ and $G$ be finite subsets of
  multi-indices such that $\norm{f-f'}_\mu,\norm{g-g'}_\mu<\epsilon/4$
  for certain $f'=\sum_{\alpha\in F} f'_\alpha e_\alpha$ and
  $g'=\sum_{\alpha\in G} g'_\alpha e_\alpha$ where $f',g'\in
  L_0^2(\Omega,\mathcal{J},\mu)$ and $\norm{f'}_\mu=\norm{g'}_\mu=1$.
  For sufficiently large $n$, we have $\abs{(f'\circ
    T^n,g')_\mu)}<\epsilon/4$. Then
  \begin{align*}
    \abs{(f\circ T^n,g)_\mu)}
    &\le
    \abs{(f'\circ T^n,g')_\mu)}
    +
    \abs{((f-f')\circ T^n,g')_\mu)}
    \\
    &\qquad
    +
    \abs{(f'\circ T^n,(g-g'))_\mu)}
    +
    \abs{((f-f')\circ T^n,(g-g')_\mu)}
    \\
    &<
    \epsilon/4
    +
    \norm{f-f'}_\mu+\norm{g-g'}_\mu
    +
    \norm{f-f'}_\mu\norm{g-g'}_\mu
    <
    \epsilon
    ,
  \end{align*}
  which completes the proof.
\end{proof}

Now let us apply this result to the generalized Santa Fe process.
A~stochastic process $(X_{i})_{i\in\mathbb{Z}}$ on
$(\Omega,\mathcal{J},P)$, where $X_i:(\Omega,\mathcal{J})\rightarrow
(\mathbb{X},\mathcal{X})$, is called 
mixing if $(\mathbb{X}^{\mathbb{Z}},\mathcal{X}^{\mathbb{Z}},\mu,T)$
is 
mixing for $\mu=P((X_k)_{k\in\mathbb{Z}}\in\cdot)$ and
$T((x_{i})_{i\in\mathbb{Z}})=(x_{i+1})_{i\in\mathbb{Z}}$.  

\begin{proof*}{Proposition \ref{theoMixing}}
  Introduce an auxiliary process $(W_i)_{i\in\mathbb{Z}}$, where
  $W_i=(K_i,(Z_{ik})_{k\in\mathbb{N}})$.  Process
  $(W_i)_{i\in\mathbb{Z}}$ is a direct product of processes
  $(K_{i})_{i\in\mathbb{Z}}$, $(Z_{i1})_{i\in\mathbb{Z}}$,
  $(Z_{i2})_{i\in\mathbb{Z}}$, ..., which are all mixing for
  $p_k\in(0,1)$. Hence $(W_i)_{i\in\mathbb{Z}}$ is mixing by
  Proposition \ref{theoStrongMixing}. (In our application, we take
  $\mu=P((W_i)_{i\in\mathbb{Z}}\in\cdot)$,
  $\mu_1=P((K_{i})_{i\in\mathbb{Z}}\in\cdot)$, and
  $\mu_{k+1}=P((Z_{ik})_{i\in\mathbb{Z}}\in\cdot)$ for $k\ge 1$. The
  transformations are
  $T((w_i)_{i\in\mathbb{Z}})=(w_{i+1})_{i\in\mathbb{Z}}$,
  $T_1((k_i)_{i\in\mathbb{Z}})=(k_{i+1})_{i\in\mathbb{Z}}$, and
  $T_{k+1}((z_i)_{i\in\mathbb{Z}})=(z_{i+1})_{i\in\mathbb{Z}}$ for
  $k\ge 1$.)  Having established the mixing property for
  $(W_i)_{i\in\mathbb{Z}}$, we notice that $X_i=f(W_i)$ for
  a~measurable function $f$. Hence $(X_{i})_{i\in\mathbb{Z}}$ is
  mixing by Theorem 3 from \cite[Chapter 10.\S
  1]{CornfeldFominSinai82en}.
\end{proof*}

\section*{Acknowledgment}

First, I~would like to thank Richard Bradley for letting me know about
his paper and Nicholas Travers for suggesting a~few other
references. Second, special thanks are due to Peter Gr\"unwald for
inviting me to Centrum Wiskunde \& Informatica, where the proof of
Proposition \ref{theoEnUDP} was drafted. Third, I~appreciate comments
of Jan Mielniczuk and the referees, regarding the paper composition.

\bibliographystyle{IEEEtran}


\end{document}